\documentclass[11pt,twoside]{article} 
\usepackage{url}
\usepackage{jeffe}
\usepackage{graphicx,hyperref}
\usepackage[utf8]{inputenc}
\usepackage[margin=1in]{geometry}
\usepackage{marvosym}
\usepackage{cite}
\usepackage{verbatim}

\newcommand{\snip}{\mathbin{\text{\raisebox{0.15ex}{\rotatebox[origin=c]{60}{\Rightscissors}\!}}}}

\DeclareMathOperator{\extend}{extend}
\DeclareMathOperator{\update}{update}
\DeclareMathOperator{\rev}{rev}
\DeclareMathOperator{\tail}{tail}
\DeclareMathOperator{\head}{head}
\DeclareMathOperator{\sig}{sig}

\newtheorem{theorem}{Theorem}[section]
\newtheorem{corollary}[theorem]{Corollary}
\newtheorem{lemma}[theorem]{Lemma}

\begin{document}

\pagestyle{myheadings}
\markboth{Minimum cycle and homology bases}
  {Glencora Borradaile, Erin Wolf Chambers, Kyle Fox, and Amir Nayyeri}

\begin{titlepage}

\title{Minimum cycle and homology bases of surface embedded graphs%
  \thanks{A preliminary version of this work was presented at the 32nd Annual International
    Symposium on Computational Geometry~\cite{bcfn-mchbs-16}.
    This material is based upon work supported by the National Science Foundation under
grants CCF-12-52833, CCF-10-54779, IIS-13-19573, CCF-11-61359, IIS-14-08846, CCF-15-13816, CCF-1617951, and
IIS-14-47554; by
an ARO grant W911NF-15-1-0408; and by Grant 2012/229 from the U.S.-Israel Binational Science
Foundation.}}

\author{
  Glencora Borradaile%
  \thanks{Oregon State University, Corvallis, OR;
    \url{{glencora,nayyeria}@eecs.orst.edu}.}
  \and
  Erin Wolf Chambers%
  \thanks{St.\ Louis University, St. Louis, MO;
    \url{echambe5@slu.edu}.}
  \and
  Kyle Fox%
  \thanks{Duke University, Durham, NC;
    \url{kylefox@cs.duke.edu}.}
  \and
  Amir Nayyeri\footnotemark[2]
}

\maketitle

\begin{abstract}
We study the problems of finding a minimum cycle basis (a minimum weight set of cycles
that form a basis for the cycle space) and a minimum homology basis (a minimum weight set of
cycles that generates the $1$-dimensional ($\Z_2$)-homology classes) of an undirected graph
embedded on a surface.
The problems are closely related, because the minimum cycle basis of a graph contains its minimum
homology basis, and the minimum homology basis of the $1$-skeleton of any graph is exactly its
minimum cycle basis.

For the minimum cycle basis problem, we give a deterministic $O(n^\omega+2^{2g}n^2+m)$-time
algorithm for graphs embedded on an orientable surface of genus~$g$.
The best known existing algorithms for surface embedded graphs are those for general graphs: an
$O(m^\omega)$ time Monte Carlo algorithm~\cite{AIJMR09} and a deterministic $O(nm^2/\log n + n^2
m)$ time algorithm~\cite{MM09}.
For the minimum homology basis problem, we give a deterministic $O((g+b)^3 n \log n + m)$-time
algorithm for graphs embedded on an orientable or non-orientable surface of genus~$g$ with~$b$
boundary components, assuming shortest paths are unique, improving on existing algorithms for many
values of~$g$ and~$n$.
The assumption of unique shortest paths can be avoided with high probability using randomization
or deterministically by increasing the running time of the homology basis algorithm by a factor
of~$O(\log n)$.
\end{abstract}

\noindent

\thispagestyle{empty}
\setcounter{page}{0}
\end{titlepage}

\section{Introduction}
\label{sec:intro}

\subsection{Minimum cycle basis}
Let~$G = (V,E)$ be a connected undirected graph with~$n$ vertices and~$m$ edges.
We define a \EMPH{cycle} of~$G$ to be a subset~$E' \subseteq E$ where each vertex~$v \in V$ is
incident to an even number of edges in~$E'$.
The \EMPH{cycle space} of~$G$ is the vector space over cycles in~$G$ where addition is defined as
the symmetric difference of cycles' edge sets.
It is well known that the cycle space of~$G$ is isomorphic to~$\Z_2^{m - n + 1}$; in particular,
the cycle space can be generated by the fundamental cycles of any spanning tree of~$G$.
A \EMPH{cycle basis} is a maximal set of independent cycles.
A \EMPH{minimum cycle basis} is a cycle basis with a minimum number of edges (counted with
multiplicity) or minimum total weight if edges are weighted%
\footnote{There is a notion of minimum cycle bases in directed graphs as well, but we focus on
the undirected case in this paper.}.
Minimum cycle bases have applications in many areas such as electrical circuit
theory~\cite{k-1847,cc-osccb-73}, structural engineering~\cite{chr-cbmms-76}, surface
reconstruction~\cite{tgg-mgpcd-06}, and the analysis of algorithms~\cite{k-acp-68}.

Sets of independent cycles form a matroid, so the minimum cycle basis can be computed using the
standard greedy algorithm.
However, there may be an exponential number of cycles in~$G$.
Horton~\cite{h-ptafs-87} gave the first polynomial time algorithm for the problem by observing
that every cycle in the minimum cycle basis is the fundamental cycle of a shortest path tree.
Several other algorithms have been proposed to compute minimum cycle bases in general
graphs~\cite{p-aspm-95,gh-ptafm-02,bgv-mcbng-04,kmmp-famcb-08,MM09,AIJMR09}.
The fastest of these algorithms are an~$O(m^\omega)$ time Monte Carlo randomized algorithm of
Amaldi \etal~\cite{AIJMR09} and an~$O(nm^2/\log n + n^2m)$ time deterministic algorithm of
Mehlhorn and Michail~\cite{MM09}.
Here,~$O(m^\omega)$ is the time it takes to multiply two $m \times m$ matrices using fast matrix
multiplication.

For the special case of \emph{planar graphs}, faster algorithms are known.
Hartvigsen and Mardon~\cite{HM94} observed that the cycles in the minimum cycle basis nest,
implicitly forming a tree; in fact, the edges of each cycle span an $s,t$-minimum cut between two
vertices in the dual graph, and the Gomory-Hu tree~\cite{gh-mtnf-61} of the dual graph is
precisely the tree of minimum cycle basis in the primal.
Hartvigsen and Mardon~\cite{HM94} gave an~$O(n^2 \log n)$ time algorithm for the minimum
cycle basis problem in planar graphs, and Amaldi \etal~\cite{AIJMR09} improved their running time
to~$O(n^2)$.
Borradaile, Sankowski, and Wulff-Nilsen~\cite{BSW15} showed how to compute an \emph{oracle} for
the minimum cycle basis and dual minimum cuts in~$O(n \log^4 n)$ time that is able to report
individual cycles or cuts in time proportional to their size.
Borradaile \etal~\cite{benw-apmcn-16} recently generalized the minimum cut oracle to graphs
embeddable on an orientable surface of genus~$g$.
Their oracle takes~$2^{O(g^2)} n \log^3 n$ time to construct (improving upon the original
planar oracle by a factor of~$\log n$).
Unfortunately, their oracle does not help in finding the minimum cycle basis in higher genus
graphs, because there is no longer a bijection between cuts in the dual graph and cycles in the
primal graph.

That said, it is not  surprising that the cycle basis oracle has not been generalized beyond
the plane.
While cuts in the dual continue to nest in higher genus surfaces, cycles do not.
In fact, the minimum cycle basis of a toroidal graph must \emph{always} contain at least one pair
of crossing cycles, because  any cycle basis must contain cycles which are topologically
distinct.
These cycles must represent different \emph{homology classes} of the surface.

\subsection{Minimum homology basis}

Given a graph~$G$ embedded in a surface~$\Sigma$ of genus~$g$ with~$b$ boundary components, the
homology of~$G$ is an algebraic description of the topology of~$\Sigma$ and of $G$'s embedding.
In this paper, we focus on one-dimensional cellular homology over the finite field~$\Z_2$.
Homology of this type allows for simplified definitions.
We say a cycle~$\eta$ is \EMPH{null-homologous} if $\eta$ is the boundary of a subset of faces.
Two cycles~$\eta$ and~$\eta'$ are \EMPH{homologous} or in the same \EMPH{homology class} if
their symmetric difference $\eta \oplus \eta'$ is null-homologous.
Let $\beta = g - \max\Set{b-1,0}$ if~$\Sigma$ is \emph{non-orientable} (contains a subset
homeomorphic to the M\"{o}bius band), and let $\beta = 2g - \max\Set{b-1, 0}$.
The homology classes form a vector space isomorphic to $\Z_2^{\beta}$ known as the \EMPH{homology
space}.
A \EMPH{homology basis} of~$G$ is a set of~$\beta$ cycles in linearly independent homology
classes, and the \EMPH{minimum homology basis} of~$G$ is the homology basis with either the
minimum number of edges or with minimum total weight if edges of $G$ are weighted.

\begin{figure}[t]
  \centering
    \includegraphics[height=0.8in]{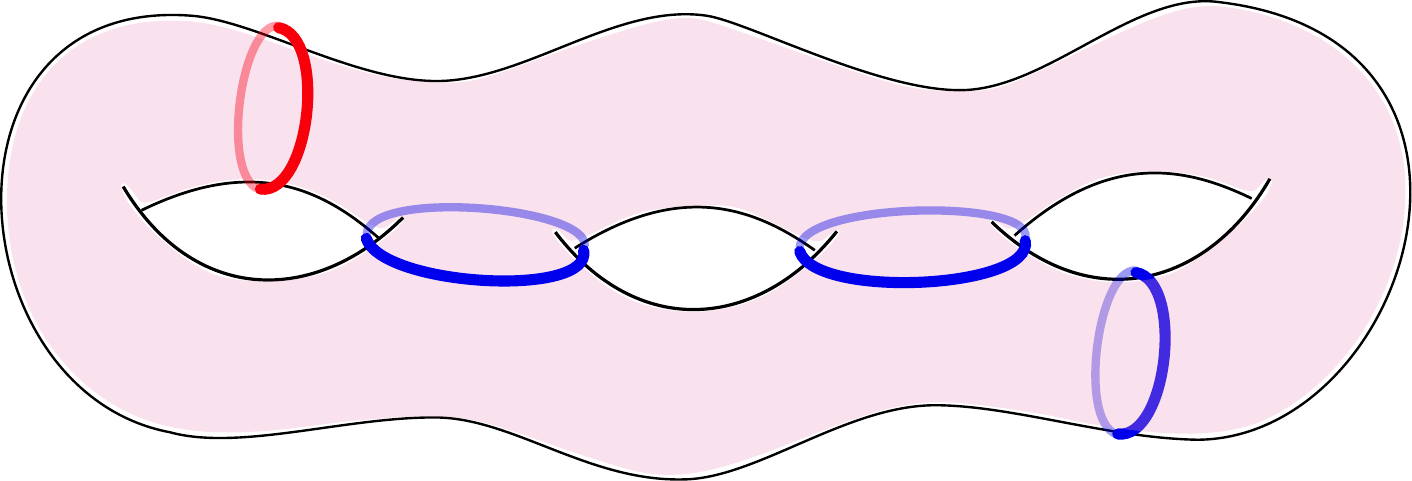}
  \caption{Two homologous cycles, one shown in red and the other in blue.}
  \label{fig:homologous}
\end{figure}

Erickson and Whittlesey~\cite{Erickson05GreedyOptGen} described an~$O(n^2 \log n + gn^2 + g^3 n)$
time algorithm for computing the minimum homology basis for orientable~$\Sigma$ without boundary.
Like Horton~\cite{h-ptafs-87}, they apply the greedy matroid basis algorithm to a set of~$O(n^2)$
candidate cycles.
Alternatively, a set of~$2^{\beta}$ candidate cycles containing the minimum homology basis can be
computed easily by applying the algorithms of Italiano \etal~\cite{insw-iamcmf-11} or Erickson and
Nayyeri~\cite{en-mcsnc-11} for computing the minimum homologous cycle in any specified homology
class.
These algorithms take~$g^{O(g)} n \log \log n$ and~$2^{O(g)} n \log n$ time respectively.
While Erickson and Whittlesey~\cite{Erickson05GreedyOptGen} do not explicitly consider the case,
all three results mentioned above can be extended to surfaces with boundary.
Similarly, the algorithms of Erickson and Whittlesey~\cite{Erickson05GreedyOptGen} and Erickson
and Nayyeri~\cite{en-mcsnc-11} can be applied to compute the minimum homology basis for
non-orientable~$\Sigma$, even though they only consider orientable surfaces explicitly.
Dey, Sun, and Wang~\cite{dsw-acsbf-10} generalized the results above to arbitrary simplicial
complexes, and Busaryev~\etal~\cite{bccdw-ashba-12} improved the running time of their
generalization from~$O(n^4)$ to~$O(n^{\omega} + n^2g^{\omega - 1})$.
Note that all of the algorithms above either take quadratic time in~$n$ (or worse) or
they have exponential dependency on~$g$.
In contrast, it is well understood how to find exactly one cycle of the minimum homology basis
of~$G$ in only~$O(g^2 n \log n)$ time assuming orientable~$\Sigma$, because the minimum weight
non-separating cycle will always be in the basis~\cite{cce-msspe-13,e-sncds-11}.

\subsection{Our results}
We describe new algorithms for computing the minimum cycle basis and minimum homology basis of the
graph~$G$.
Our algorithm for minimum cycle basis requires~$G$ be embedded on an orientable~$\Sigma$, but it
is deterministic and runs in $O(n^\omega + 2^{2g} n^2 + m)$ time, matching the running time of the
randomized algorithm of Amaldi \etal~\cite{AIJMR09} when~$g$ is sufficiently small.
Our algorithm for minimum homology basis works in orientable or non-orientable~$\Sigma$,  is also
deterministic, and it runs in $O((g+b)^3 n \log n + m)$ time assuming shortest paths are unique.
The assumption of unique shortest paths is only necessary to use the multiple-source shortest path
data structure of Cabello, Chambers, and Erickson~\cite{cce-msspe-13}.
It can be avoided with high probability by using randomization or deterministically by increasing
the running time of our algorithm by a factor of~$O(\log n)$~\cite{cce-msspe-13}.  For simplicity,
we will assume shortest paths are unique during the presentation of our minimum homology basis
algorithm.
In any case, ours is the first algorithm for minimum homology basis that has a running time
simultaneously near-linear in~$n$ and polynomial in~$g$.

At a high level, both of our algorithms are based on the~$O(n m^2 + n^2 m \log n)$ time algorithm
of Kavitha \etal~\cite{kmmp-famcb-08} who in turn use an idea of de Pina~\cite{p-aspm-95}.
We compute our basis cycles one by one.
Over the course of the algorithm, we maintain a set of \emph{support vectors} that form the basis
of the subspace that is \emph{orthogonal} to the set of cycles we have already computed.
Every time we compute a new cycle, we find one of minimum weight that is \emph{not} orthogonal
to a chosen support vector~$S$, and then update the remaining support vectors so they remain
orthogonal to our now larger set of cycles.
Using the divide-and-conquer approach of Kavitha \etal~\cite{kmmp-famcb-08}, we are able to
maintain these support vectors in only~$O(n^\omega)$ time total in our minimum cycle basis
algorithm and~$O(g^\omega)$ time total in our minimum homology basis algorithm.
Our approaches for picking the minimum weight cycle not orthogonal to~$S$ form the more
technically interesting parts of our algorithms and are unique to this work.

For our minimum cycle basis algorithm, we compute a collection of~$O(2^{2g}n)$ cycles that contain
the minimum cycle basis and then partition these cycles according to their homology classes.
The cycles within a single homology class nest in a similar fashion to the minimum cycle basis
cycles of a planar graph.
Every time we compute a new cycle for our minimum cycle basis, we walk up the $2^{2g}$ trees of
nested cycles and find the minimum weight cycle not orthogonal to~$S$ in~$O(n)$ time per tree.
Overall, we spend~$O(2^{2g} n^2)$ time finding these cycles;
if any improvement is made on the time it takes to update the support vectors, then the running
time of our algorithm as a whole will improve as well.

Our minimum homology basis algorithm uses a covering space called the cyclic double cover.
As shown by Erickson~\cite{e-sncds-11}, the cyclic double cover provides a convenient way to find
a minimum weight closed walk~$\gamma$ crossing an arbitrary non-separating cycle~$\lambda$ an odd
number of times.
We extend his construction so that we may consider not just one~$\lambda$ but any arbitrarily
large collection of cycles.
Every time we compute a new cycle in our minimum homology basis algorithm, we let~$S$ determine a
set of cycles that must be crossed an odd number of times, build the cyclic double cover for that
set, and then compute our homology basis cycle in~$O((g+b) g n \log n)$ time by computing minimum
weight paths in the covering space%
\footnote{In addition to the above results, we note that it is possible to improve the~$g^{O(g)} n
  \log \log n$ time algorithm for minimum homology basis based on Italiano
  \etal~\cite{insw-iamcmf-11} so that it runs in $2^{O(g)} n \log \log n$ time.
  However, this improvement is a trivial adaption of techniques used by Fox~\cite{f-sncdu-13} to
  get a $2^{O(g)} n \log \log n$ time algorithm for minimum weight non-separating and
  non-contractible cycle in undirected graphs.
We will not further discuss this improvement in our paper.}.

The rest of the paper is organized as follows.
We provide more preliminary material on surface embedded graphs in Section~\ref{sec:prelims}.
In Section~\ref{sec:signatures}, we describe a characterization of cycles and homology classes
using binary vectors.
These vectors are helpful in formally defining our support vectors.
We give a high level overview of our minimum cycle basis algorithm in
Section~\ref{sec:cycle_basis} and describe how to pick individual cycles in
Section~\ref{sec:selecting-cycles}.
Finally, we give our minimum homology basis algorithm in Section~\ref{sec:homology_basis}.

\section{Preliminaries}
\label{sec:prelims}
We begin with an overview of graph embeddings on surfaces.
For more background, we refer readers to books and surveys on topology~\cite{h-at-02,m-t-00},
computational topology~\cite{eh-cti-10,z-tc-05}, and graphs on surfaces~\cite{c-tags-12,mt-gs-01}.

A \EMPH{surface} or 2-manifold with boundary~$\Sigma$ is a compact Hausdorff space in which every
point lies in an open neighborhood homeomorphic to the Euclidean plane or the closed half plane.
The \EMPH{boundary} of the surface is the set of all points whose open neighborhoods are
homeomorphic to the closed half plane.
Every boundary component is homeomorphic to the circle.
A \EMPH{cycle} in the surface~$\Sigma$ is a continuous function $\gamma : S^1 \to \Sigma$,
where~$S^1$ is the unit circle.
Cycle~$\gamma$ is called \EMPH{simple} if~$\gamma$ is injective.
A \EMPH{path} in surface~$\Sigma$ is a continuous function $p : [0,1] \to \Sigma$; again, path~$p$
is simple if~$p$ is injective.
A \EMPH{loop} is a path $p$ such that $p(0) = p(1)$; equivalently, it is a cycle with a designated
basepoint.
The \EMPH{genus} of the surface~$\Sigma$, denoted by~$g$, is the maximum number of disjoint simple
cycles~$\gamma_1, \dots, \gamma_g$ in~$\Sigma$ such that $\Sigma \setminus (\gamma_1 \cup \dots
\cup \gamma_g)$ is connected.
Surface~$\Sigma$ is \EMPH{non-orientable} if it contains a subset homeomorphic to the M\"{o}bius
band; otherwise, it is \EMPH{orientable}.

The \EMPH{embedding} of graph~$G = (V,E)$ is a drawing of~$G$ on~$\Sigma$ which maps vertices to
distinct points on~$\Sigma$ and edges to internally disjoint simple paths whose endpoints lie on
their incident vertices' points.
A \EMPH{face} of the embedding is a maximally connected subset of~$\Sigma$ that does not intersect
the image of~$G$.
An embedding is \EMPH{cellular} if every face is homeomorphic to an open disc; in particular,
every boundary component must be covered by (the image of) a cycle in~$G$.
These boundary cycles must be vertex-disjoint.
Embeddings can be described combinatorially using a rotational system and a signature.
The \EMPH{rotation system} describes the cyclic ordering of edges around each vertex.
The \EMPH{orientation signature} $\sig : E \to \Set{0,1}$ is a function that assigns to each
edge~$e$ a bit.
Value $\sig(e) = 0$ if the cyclic ordering of $e$'s endpoints are in the same direction;
otherwise, $\sig(e) = 1$.
Abusing notation, we denote the orientation signature of a cycle~$\eta$ (in~$G$) as~$\sig(\eta)$
and define it as the exclusive-or of its edges' orientation signatures.
If~$\sig(\eta) = 1$, we say~$\eta$ is \EMPH{one-sided}.
Otherwise, we say $\eta$ is \EMPH{two-sided}.
Surface~$\Sigma$ is orientable if and only if every cycle of~$G$ is two-sided.

Let~$F$ be the set of faces in~$G$.
Let~$n$, $m$, $\ell$, and $b$ be the number of vertices, edges, faces, and boundary components
of~$G$'s embedding respectively.
The \EMPH{Euler characteristic~$\chi$} of~$\Sigma$ is $2-2g - b$ if~$\Sigma$ is orientable and is
$2 - g - b$ otherwise.
By Euler's formula, $\chi = n - m + \ell$.
Embedded graphs can be \EMPH{dualized}: $G^*$ is the graph embedded on the same surface, with a
vertex in $G^*$ for every face \emph{and boundary component} in $G$ and a face in $G^*$ for every
vertex of $G$.
We refer to the dual vertices of boundary components as \EMPH{boundary dual vertices}.
Two vertices in $G^*$ are adjacent if the corresponding faces/boundary components are separated by
an edge in $G$.
We generally do not distinguish between edges in the primal and dual graphs.
We assume~$\Sigma$ contains at least one boundary component as it does not affect the
homology of~$\Sigma$ to remove a face when there are no boundary to begin with.
In particular, this assumptions simplifies the definition of~$\beta$ as given in the introduction
so $\beta = g-b+1$ if~$\Sigma$ is non-orientable and~$\beta = 2g-b+1$ otherwise.

A \EMPH{spanning tree} of the graph $G$ is a subset of edges of $G$ which form a tree containing
every vertex.
A \EMPH{spanning coforest} is a subset of edges which form a forest in the dual graph with exactly
$b$ components, each containing one dual boundary vertex.
A \EMPH{tree-coforest decomposition} of $G$ is a partition of $G$ into 3 edge disjoint subsets,
$(T,L,C)$, where $T$ is a spanning tree of $G$, $C$ is a spanning coforest, and $L$ is the set of
leftover edges $E \setminus (T \cup C)$~\cite{e-dgteg-03,en-mcsnc-11}.
Euler's formula implies $\abs{L} = \beta$.

A \EMPH{$w,w'$-path} $p$ (in~$G$) is an ordered sequence of edges $\Set{u_1 v_1, u_2
v_2,\dots,u_{k}v_{k}}$ where $w=u_1$, $w'=v_k$, and~$v_i = u_{i+1}$ for all positive $i < k$; a
\EMPH{closed path} is a path which starts and ends on the same vertex.
A path is \EMPH{simple} if it repeats no vertices (except possibly the first and last).
We sometimes use \EMPH{simple cycle} to mean a simple closed path.
A path in the dual graph~$G^*$ is referred to as a \EMPH{co-path} and a cycle in the dual is
referred to as a \EMPH{co-cycle}.
Simple paths and cycles in the dual are referred to as simple co-paths and co-cycles respectively.
Every member of the minimum cycle basis (and subsequently the minimum homology basis) is a simple
cycle~\cite{h-ptafs-87}.
We let~$\sigma(u,v)$ denote an arbitrary shortest (minimum weight) $u,v$-path in~$G$.
Let~$p[u,v]$ denote the \EMPH{subpath} of~$p$ from~$u$ to~$v$.
Given a $u,v$-path $p$ and a $v,w$-path $p'$, let $p \cdot p'$ denote their concatenation.
Two paths $p$ and $p'$ \EMPH{cross} if their embeddings in~$\Sigma$ cannot be be made disjoint
through infinitesimal perturbations; more formally, they cross if there is a maximal (possibly
trivial) common subpath $p''$ of $p$ and $p'$ such that, upon contracting $p''$ to a vertex $v$,
two edges each of $p$ and $p'$ alternate in their embedded around $v$.
Two closed paths cross if they have subpaths which cross.

Let $\gamma$ be a closed path in~$G$ that does not cross itself.
We define the operation of \EMPH{cutting} along $\gamma$ and denote it $G \snip \gamma$.
Graph $G \snip \gamma$ is obtained by cutting along $\gamma$ in the drawing of $G$ on the
surface, creating two copies of $\gamma$.
If~$\sig(\gamma) = 0$, then the two copies of $\gamma$ each form boundary components in the cut
open surface.
Otherwise, the two copies of $\gamma$ together form a single closed path that is the concatenation
of $\gamma$ to itself at both ends; the single closed path forms a single boundary component.
Likewise, given a simple path~$\sigma$ in~$G$, we obtain the graph $G \snip \sigma$ by cutting
along $\sigma$, creating two interiorly disjoint copies connected at their endpoints.
The cut open surface has one new boundary component bounded by the copies of~$\sigma$.

Let~$F'$ be a collection of faces and boundary components.
Let~$\partial F'$ denote the boundary of~$F'$, the set of edges with exactly one incident face or
boundary component in~$F'$.
We sometimes call~$F'$ a \EMPH{cut} of $G^*$ and say $\partial F'$ \EMPH{spans} the cut.
A co-path $p$ with edge $uv \in \partial F'$ \EMPH{crosses} the cut at $uv$.

Finally, let~$w$ and~$w'$ be two bit-vectors of the same length.
We let~$\langle w , w' \rangle$ denote the \EMPH{dot product} of $w$ and $w'$, defined by the
exclusive-or of the products of their corresponding bits.

\subsection{Sparsifying~$G$}
We assume~$g = O(n^{1-\varepsilon})$ for some constant $\varepsilon > 0$; otherwise, our
algorithms offer no improvement over previously known results.
Because boundary cycles are vertex-disjoint, we also have~$b = O(n)$.
Here, we describe a procedure to preprocess~$G$ so it contains a linear (in~$n$) number of edges
and faces as well.
Namely, we modify~$G$ so it contains no faces of degree $2$ or $1$; Euler's formula immediately
limits the number of edges and faces in the modified graph to~$O(n)$.
We begin by computing all pairs shortest paths.
Observe, if we simply remove parallel edges from~$G$, then we only reduce the minimum genus of any
surface into which we can embed~$G$.
In particular, the graph would have~$O(n)$ edges left after removing all parallel edges.
All pairs shortest paths can be computed in $O(n^2 \log n + m)$ time.

Now, consider the graph~$G$ as given in the input.
We iteratively perform the following procedure until every face has degree $3$ or greater or our
graph is one of a constant number of easy cases.
In each iteration, we add at most one cycle to the minimum cycle basis or minimum homology basis.
Let~$f$ be a face of degree~$1$ or~$2$.
If~$f$ has degree~$1$, then it is bounded by a null-homologous loop~$e$ in~$G$.
We add $\{e\}$ to the minimum cycle basis, because it is the cheapest cycle containing~$e$, but we
do not add it to the minimum homology basis.
If~$G$ consists entirely of~$e$, we terminate; otherwise we remove~$e$ and~$f$ from the graph and
continue with the next iteration.
If~$f$ has degree~$2$, then it is either bounded by two distinct edges~$e$ and~$e'$ or bounded
twice by a single edge.
In the latter case, graph~$G$ must be the path of length~$1$ embedded in the sphere or it is a
single vertex and non-null homologous loop embedded in the projective plane (the non-orientable
surface of genus~$1$).
If it is the path in the sphere, we add nothing to the minimum cycle basis and minimum homology
basis, and we terminate.
If it is a loop in the projective plane, we add it to both the minimum cycle basis and the minimum
homology basis and terminate.
Now suppose~$f$ is bounded by distinct faces~$e$ and~$e'$, and let~$e$ have less weight than~$e'$
without loss of generality.
Edge~$e'$ belongs to cycle~$\{e,e'\}$, so it belongs to some cycle of the minimum cycle basis.
Let~$\sigma$ be the shortest path between the endpoints of~$e'$.
We add~$\sigma \cdot e'$ to the minimum cycle basis.
No other cycle in the minimum cycle basis contains~$e'$, because it would always be at least as
cheap to include~$e$ in the cycle instead.
Also, no cycle of the minimum homology basis contains~$e'$, because it would always be at least as
cheap to include~$e$ and $\{e,e'\}$ itself is null-homologous.
We remove~$e'$ and~$f$ from the graph and continue with the next iteration.

Each iteration is done in constant time, and there are at most~$m$ iterations of the above
algorithm.
Therefore, the preprocessing procedure takes $O(n^2 \log n + m)$ time total.
We assume for the rest of that paper that $m = O(n)$ and $\ell = O(n)$.

\section{Cycle and Homology Signatures}
\label{sec:signatures}

We begin the presentation of our algorithms by giving a characterization of cycles and homology
classes using binary vectors.
These vectors will be useful in helping us determine which cycles can be safely added to our
minimum cycle and homology bases.
Let~$(T, L, C)$ be an arbitrary tree, coforest decomposition of~$G$;
set~$L$ contains exactly~$\beta$ edges $e_1, \dots, e_{\beta}$.
For each index~$i \in \Set{1, \dots, \beta}$, graph~$C \cup \Set{e_i}$ contains a unique simple
co-cycle or a unique simple co-path between distinct dual boundary vertices.
Let~$p_i$ denote this simple co-cycle or co-path.
Let~$f_{\beta+1}, \dots, f_{m - n + 1}$ denote the $m - n + 1 - \beta = \ell$ faces of~$G$, and
for each index $i \in \Set{\beta+1, \dots, m - n + 1}$, let~$p_i$ denote the simple co-path
from~$f_i$ to the dual boundary vertex in $f_i$'s component of~$C$.

For each edge~$e$ in~$G$, we define its \EMPH{cycle signature}~$[e]$ as an $(m - n + 1)$-bit
vector whose $i$th bit is equal to~$1$ if and only if $e$ appears in~$p_i$.
The cycle signature~$[\eta]$ of any cycle~$\eta$ is the bitwise exclusive-or of the
signatures of its edges.
Equivalently, the $i$th bit of~$[\eta]$ is~$1$ if and only if~$\eta$ and~$p_i$ share an odd number
of edges.
Similarly, for each edge~$e$ in~$G$, we define its \EMPH{homology signature}~$[e]_h$ as a
$\beta$-bit vector whose $i$th bit is equal to~$1$ if and only if $e$ appears in~$p_i$.
The homology signature of cycles is defined similarly.

The following lemma is immediate.
\begin{lemma}
  \label{lem:signature_sums}
  Let~$\eta$ and~$\eta'$ be two cycles.
  We have $[\eta \oplus \eta'] = [\eta] \oplus [\eta']$ and $[\eta \oplus \eta']_h = [\eta]_h \oplus
  [\eta']_h$.
\end{lemma}

Let~$\zeta_i$ denote the unique simple cycle in~$T \cup \{e_i\}$.
The following lemma helps us explain the properties of cycle and homology signatures.
\begin{lemma}
  \label{lem:easy_basis}
  The set of cycles~$\Set{\zeta_1, \dots, \zeta_\beta}$ form a homology basis.
\end{lemma}
\begin{proof}
  We prove that the cycles lie in independent homology classes by showing the symmetric difference
  of any non-empty subset of~$\Set{\zeta_1, \dots, \zeta_\beta}$ is not null-homologous.
  Suppose to the contrary that there exists a non-empty $\Upsilon \subseteq \Set{\zeta_1, \dots,
  \zeta_\beta}$ such that $\bigoplus_{\eta \in \Upsilon} \eta = \partial F'$ for some subset of
  faces~$F' \subseteq F$, where~$\bigoplus$ is the symmetric difference of its operands.
  Let~$\zeta_i \in \Upsilon$ be an arbitrary member of the subset.
  Co-path~$p_i$ shares exactly one edge with~$\zeta_i$, and it shares no edges with any
  other~$\eta \in \Upsilon$.
  In particular,~$p_i$ crosses dual cut~$F'$ an odd number of times.
  Therefore,~$p_i$ cannot be a co-cycle.
  Further,~$p_i$ cannot be a co-path between two distinct dual boundary vertices, because exactly
  one of those two vertices would have to lie inside~$F'$, a contradiction on~$F'$ only containing
  faces.
  We conclude~$\Upsilon$ cannot exist and the cycles~$\Set{\zeta_1, \dots, \zeta_\beta}$ do lie in
  independent homology classes.
\end{proof}
Let~$w$ be an arbitrary $(m - n + 1)$-bit vector.
We construct a cycle~$\eta_w$ to demonstrate how cycle and homology signatures provide a
convenient way to distinguish between cycles and their homology classes.
Let~$\Upsilon \subseteq \Set{\zeta_1, \dots, \zeta_\beta}$ be the subset of basis cycles
containing exactly the cycles~$\zeta_i$ such that the $i$th bit of~$w$ is equal to~$1$.
Similarly, let~$F' \subseteq F$ be the subset of faces such that face~$f_i \in F'$ if and only
the $i$th bit of~$w$ is equal to~$1$.
Let~$\eta_w = \bigoplus_{\eta \in (\Upsilon \cup \{\partial F'\})} \eta$.
\begin{lemma}
  \label{lem:bijection}
  We have~$[\eta_w] = w$.
\end{lemma}
\begin{proof}
  Let~$i \in \Set{1,\dots,m - n + 1}$, and let~$p_i$ be the co-path as defined above.
  Suppose~$i \in \Set{1,\dots,\beta}$.
  Co-path~$p_i$ crosses cut~$F'$ an even number of times.
  If bit~$i$ in~$w$ is set to~$1$, then $p_i$ shares exactly one edge of~$\bigoplus_{\eta \in
  \Upsilon} \eta$ by construction, and it must share an odd number of edges with~$\eta_w$ as well.
  If bit~$i$ in~$w$ is set to~$0$, then $p_i$ shares no edges with $\bigoplus_{\eta \in \Upsilon}
  \eta$, and it must share an even number of edges with~$\eta_w$.

  Now, suppose~$i \in \Set{\beta+1,\dots,m-n+1}$.
  Co-path~$p_i$ shares no edges with $\bigoplus_{\eta \in \Upsilon} \eta$.
  If~$i$ is set to~$1$, then~$f_i \in F'$ and~$p_i$ crosses cut~$F'$ an odd number of times.
  Therefore, it shares an odd number of edges with~$\eta_w$.
  If~$i$ is set to~$0$, then~$f_i \notin F'$, and~$p_i$ crosses cut~$F'$ an even number of times,
  sharing an even number of edges with~$\eta_w$.
\end{proof}

\begin{corollary}
  \label{cor:distinct_cycle_signatures}
  Let~$\eta$ and~$\eta'$ be two cycles.
  We have~$\eta = \eta'$ if and only if~$[\eta] = [\eta']$.
\end{corollary}
Observe that the homology class of~$\eta_w$ is entirely determined by the first~$\beta$ bits
of~$w$.
We immediately obtain an alternative (and more combinatorially inspired) proof of the following
corollary of Erickson and Nayyeri~\cite{en-mcsnc-11}.
\begin{corollary}[Erickson and Nayyeri~{\cite[Corollary 3.3]{en-mcsnc-11}}]
  \label{cor:distinct_homology_signatures}
  Two cycles~$\eta$ and $\eta'$ are homologous if and only if~$[\eta]_h = [\eta']_h$.
\end{corollary}
\begin{corollary}
  \label{cor:signature_isomorphisms}
  Cycle signatures are an isomorphism between the cycle space and~$\Z_2^{m - n + 1}$, and
  homology signatures are an isomorphism between the first homology space and~$\Z_2^{2g}$.
\end{corollary}

\section{Minimum Cycle Basis}
\label{sec:cycle_basis}
We now describe our algorithm for computing a minimum cycle basis.
We assume without loss of generality that surface~$\Sigma$ contains exactly one boundary
component, because the addition or removal of boundary does not affect the cycles of~$G$.
We denote the one boundary component and its corresponding dual boundary vertex as $f_\infty$,
because it is analogous to the infinite face of a planar graph.
Our algorithm for minimum cycle basis is only for~$G$ embedded on an orientable surface~$\Sigma$.
We conclude~$\beta = 2g$.

Our algorithm is based on one of Kavitha, Mehlhorn, Michail and Paluch~\cite{kmmp-famcb-08} which
is in turn based on an algorithm of de Pina~\cite{p-aspm-95}.
Our algorithm incrementally adds simple cycles~$\gamma_1, \dots, \gamma_{m - n + 1}$ to the
minimum cycle basis.
In order to do so, it maintains a set of $(m - n + 1)$-bit \EMPH{support vectors}~$S_1, \dots,
S_{m - n + 1}$ with the following properties:
\begin{itemize}
  \item
    The support vectors form a basis for~$\Z_2^{m - n + 1}$.
  \item
    When the algorithm is about to compute the $j$th simple cycle~$\gamma_j$ for the minimum cycle
    basis, $\langle S_j, [\gamma_{j'}] \rangle = 0$ for all~$j' < j$.
\end{itemize}
Our algorithm chooses for~$\gamma_j$ the minimum weight cycle~$\gamma$ such that~$\langle S_j,
[\gamma] \rangle = 1$.
Note that~$S_j$ must have at least one bit set to~$1$, because the set of vectors~$S_1, \dots,
S_{m - n + 1}$ forms a basis.
Therefore, such a~$\gamma$ does exist; in particular, we could choose~$[\gamma]$ to contain
exactly one bit equal to~$1$ which matches any $1$-bit of~$S_j$.
The correctness of choosing~$\gamma_j$ as above is guaranteed by the following lemma.
\begin{lemma}
  \label{lem:min_weight_odd_hit}
  Let~$S$ be an $(m - n + 1)$-bit vector with at least one bit set to~$1$, and let~$\eta$ be the
  minimum weight cycle such that~$\langle S, [\eta] \rangle = 1$.
  Then,~$\eta$ is a member of the minimum cycle basis.
\end{lemma}
\begin{proof}
  Let~$\eta_1, \dots, \eta_{2^{m - n + 1}}$ be the collection of cycles ordered by increasing
  weight, and choose~$j$ such that~$\eta_j = \eta$.
  For any subset~$\Upsilon$ of~$\Set{\eta_1, \dots, \eta_{j-1}}$, we have~$\langle [\bigoplus_{\eta'
  \in \Upsilon} \eta'], S \rangle = 0$.
  Therefore,~$\eta$ is independent of~$\Set{\eta_1, \dots, \eta_{j-1}}$.
  Sets of independent cycles form a matroid, so~$\eta$ must be a member of the minimum weight
  cycle basis.
\end{proof}

\subsection{Maintaining support vectors}
\label{subsec:cycle_basis_algorithm}
Our algorithm updates the support vectors and computes minimum cycle basis vectors in a recursive
manner.
Initially, each support vector~$S_i$ has only its $i$th bit set to~$1$.
Borrowing nomenclature from Kavitha~\etal~\cite{kmmp-famcb-08}, we define two procedures,
$\extend(j, k)$ which extends the current set of basis cycles by adding~$k$ cycles starting
with~$\gamma_j$, and $\update(j, k)$ which updates support vectors~$S_{j + \lfloor k/2 \rfloor},
\dots, S_{j + k - 1}$ so that for any~$j', j''$ with $j + \lfloor k/2 \rfloor \leq j' < j + k$
and~$1 \leq j'' < j + \lfloor k/2 \rfloor$, we have~$\langle S_{j'}, [\gamma_{j''}] \rangle = 0$.
Our algorithm runs $\extend(1, m - n + 1)$ to compute the minimum cycle basis.

We implement $\extend(j,k)$ in the following manner:
If~$k > 1$, then our algorithm recursively calls~$\extend(j, \lfloor k/2 \rfloor)$ to add~$\lfloor
k/2 \rfloor$ cycles to the partial minimum cycle basis.
It then calls $\update(j,k)$ so that support vectors $S_{j + \lfloor k/2 \rfloor}, \dots, S_{j + k
- 1}$ become orthogonal to the newly added cycles of the partial basis.
Finally, it computes the remaining~$\lceil k/2 \rceil$ basis cycles by calling $\extend(j +
\lfloor k/2 \rfloor, \lceil k/2 \rceil)$.
If~$k = 1$, then $\langle S_j, [\gamma_{j'}] \rangle = 0$ for all~$j' < j$.
Our algorithm is ready to find basis cycle~$\gamma_j$.
We describe an~$O(2^{2g} n)$ time procedure to find~$\gamma_j$ in Section~\ref{sec:selecting-cycles}.

We now describe $\update(j,k)$ in more detail.
Our algorithm updates each support vector~$S_{j'}$ where $j + \lfloor k/2 \rfloor \leq j' < j +
k$.
The vector~$S_{j'}$ becomes~$S'_{j'} = S_{j'} + \alpha_{j'0} S_{j} + \alpha_{j'1} S_{j+1} + \cdots
+ \alpha_{j' (\lfloor k/2 \rfloor - 1)} S_{j + \lfloor k/2 \rfloor - 1}$ for some set of scalar
bits~$\alpha_{j' 0} \dots \alpha_{j' (\lfloor k/2 \rfloor - 1)}$.
After updating, the set of vectors $S_1, \dots, S_{m - n + 1}$ remains a basis for~$\Z_2^{m - n +
1}$ regardless of the choices for the~$\alpha$ bits.
Further,~$\langle S'_{j'}, [\gamma_{j''}] \rangle = 0$ for all~$j'' < j$ for all choices of
the~$\alpha$ bits, because $\extend(j,k)$ is only called after its support vectors are updated to
be orthogonal to all minimum basis cycles~$\gamma_1, \dots, \gamma_{j-1}$.

However, it is non-trivial to guarantee~$\langle S'_{j'}, [\gamma_{j''}] \rangle = 0$ for
all~$j''$ where $j \leq j'' < j + \lfloor k/2 \rfloor$.
Let~$w^T$ denote the transpose of a vector~$w$.
Let
$$ X =
\begin{pmatrix}
  S_j \\
  \cdots \\
  S_{j+\lfloor k/2 \rfloor - 1}
\end{pmatrix}
\cdot \Paren{[\gamma_j]^T \cdots [\gamma_{j + \lfloor k/2 \rfloor - 1}]^T}$$
and
$$ Y =
\begin{pmatrix}
  S_{j + \lfloor k/2 \rfloor} \\
  \cdots \\
  S_{j + k - 1}
\end{pmatrix}
\cdot \Paren{[\gamma_j]^T \cdots [\gamma_{j + \lfloor k/2 \rfloor - 1}]^T}.$$
Let $A = YX^{-1}$.
Row $j' - j - \lfloor k/2 \rfloor + 1$ of matrix~$A$ contains exactly the bits~$\alpha_{j' 0}
\dots \alpha_{j' (\lfloor k/2 \rfloor - 1)}$ we are seeking~\cite[Section 4]{kmmp-famcb-08}.
Matrices~$X$,~$Y$, and~$A$ can be computed in~$O(n k^{\omega - 1})$ time using fast matrix
multiplication and inversion, implying the new support vectors~$S'_{j+ \lfloor k/2 \rfloor},
\dots, S'_{j+k-1}$ can be computed in the same amount of time.

We can bound the running time of $\extend(j,k)$ using the following recurrence:
$$
  T(k) =
  \begin{cases}
    2T(k/2) + O(nk^{\omega - 1}) & \text{if } k > 1 \\
    O(2^{2g} n) & \text{if } k = 1
  \end{cases}
$$

The total time spent in calls to $\extend(j,k)$ where~$k > 1$ is~$O(n k^{\omega - 1})$.
The total time spent in calls to $\extend(j,1)$ is~$O(2^{2g}nk)$.
Therefore,~$T(k) = O(nk^{\omega - 1} + 2^{2g}nk)$.
The running time of our minimum cycle basis algorithm (after sparsifying~$G$) is~$T(O(n)) =
O(n^{\omega} + 2^{2g}n^2)$.

\section{Selecting Cycles}
\label{sec:selecting-cycles}

A {\em Horton cycle} is a simple cycle given by a shortest $x,u$-path, a shortest $x,v$-path, and
the edge $uv$; in particular, the set of all Horton cycles is given by the set of $m - n + 1$
elementary cycles for each of the $n$ shortest path trees~\cite{h-ptafs-87}.
Thus, in sparse graphs, there are $O(n^2)$ Horton cycles.
A simple cycle $\gamma$ of a graph $G$ is {\em isometric} if for every pair of vertices $x, y \in
\gamma$, $\gamma$ contains a shortest $x,y$-path.
Hartvigsen and Mardon prove that the cycles of any minimum cycle basis are all
isometric~\cite{HM94}.
Therefore, it suffices for us to focus on the set of isometric cycles to find the cycle $\gamma_j$ as needed for Section~\ref{subsec:cycle_basis_algorithm}.
Amaldi \etal~\cite{AIJMR09} show how to extract the set of distinct isometric cycles from a set
of Horton cycles in $O(nm)$ time.
Each isometric cycle is identified by a shortest path tree's root and a non-tree edge.

Here, we show that there are at most $O(2^{2g}n)$ isometric cycles in our graph of genus
$g$ (Section~\ref{sec:isom-cycl-surf}), and they can be partitioned into sets according to their
homology classes.
We can represent the isometric cycles in a given homology class using a tree that can be built in
$O(n^2)$ time (Section~\ref{sec:repr-isom-cycl}).
We then show that we can use these trees to find the minimum-cost cycle $\gamma_j$ as needed for
Section~\ref{subsec:cycle_basis_algorithm} in linear time per homology class of isometric cycles.
We close with a discussion on how to improve the running time for computing and representing
isometric cycles (Section~\ref{sec:impr-time-find}).
While these improvements do not improve the overall running time of our algorithm (by maintaining separate representations of the cycles according to their homology class, we require linear time per representation to process the support vector with respect to which $\gamma_j$ is non-orthogonal; we also require
$O(n^\omega)$ time to update the support vectors), it does further emphasize the bottleneck our
algorithm faces in updating and representing the support vectors.

\subsection{Isometric cycles in orientable surface embedded graphs}\label{sec:isom-cycl-surf}

Here we prove some additional structural properties that isometric cycles have in orientable
surface embedded graphs.
To this end, we herein assume that shortest paths are unique.
Hartvigsen and Mardon show how to achieve this assumption algorithmically when, in particular, all pairs of shortest paths are computed, as we do~\cite{HM94}.
We first prove a generalization of the following lemma for the planar case by Borradaile, Sankowski and Wulff-Nilsen.

\begin{lemma}[Borradaile {\etal~\cite[Lemma 1.4]{BSW15}}]\label{lem:simple-int}
  Let $G$ be a graph in which shortest paths are unique. The intersection between an isometric
  cycle and a shortest path in $G$ is a (possibly empty) shortest path.  The intersection between
  two distinct isometric cycles $\gamma$ and $\gamma'$ in $G$ is a (possibly empty) shortest path;
  in particular, if $G$ is a planar embedded graph, $\gamma$ and $\gamma'$ do not cross.
\end{lemma}

\begin{lemma}\label{lem:nestingisom}
  Two isometric cycles in a given homology class in a graph with unique shortest paths do not cross.
\end{lemma}

\begin{proof}
  Let $\gamma$ and $\gamma'$ be two isometric cycles in a given homology class.  Suppose for a
  contradiction that $\gamma$ and $\gamma'$ cross.  By the first part of
  Lemma~\ref{lem:simple-int}, and the assumption that $\gamma$ and $\gamma'$ cross, $\gamma \cap
  \gamma'$ is a single simple path $p$.  Therefore, $\gamma$ and $\gamma'$ cross exactly once.

Suppose $\gamma$ and $\gamma'$ are not null-homologous.
Cutting the surface open along $\gamma$ results in a connected surface with two boundary
components which are
connected by $\gamma'$.  Cutting the surface further along $\gamma'$ does not disconnect the
surface.  Therefore $\gamma \oplus \gamma'$ does not disconnect the surface, and so $\gamma$ and
$\gamma'$ are not homologous, a contradiction.

If $\gamma$ and $\gamma'$ are null-homologous, then cutting the surface open along $\gamma$
results in a disconnected surface in which $\gamma' \setminus p$ is a path, but between different components of the surface, a contradiction.
\end{proof}

\begin{corollary}\label{cor:isom}
There are at most $\ell$ distinct isometric cycles in a given homology class in a graph with
$\ell$ faces and unique shortest paths.
\end{corollary}
\begin{proof}
  Consider the set $\Set{C_1, C_2, \ldots}$ of distinct isometric cycles in a
  given homology class other than the null homology class.
  We prove by induction that $\Set{C_1, C_2, \ldots, C_i}$
  cut the surface into non-trivial components, each of which is bounded by exactly two of
  $C_1, C_2, \ldots, C_i$; this is true for $C_1, C_2$ since they are
  homologous, distinct and do cross.
  $C_{i+1}$ must be contained in one component, bounded
  by, say, $C_j$ and $C_k$ since $C_{i+1}$ does not cross any other cycle.
  Cutting
  this component along $C_{i+1}$ creates two components bounded by $C_j,C_{i+1}$
  and $C_k, C_{i+1}$, respectively.
  Since the cycles are distinct, these component must
  each contain at least one face.
  A similar argument holds for the set of null-homologous isometric cycles.
\end{proof}

Since there are $2^{2g}$ homology classes and $\ell = O(n)$, we get:
\begin{corollary} \label{cor:superset}
  There are $O(2^{2g}n)$ distinct isometric cycles in a graph of orientable genus $g$ with unique
  shortest paths.  
\end{corollary}

We remark that Lemma~\ref{lem:nestingisom} is not true for graphs embedded in non-orientable
surfaces, because homologous cycles may cross exactly once.
In fact, one can easily construct an arbitrarily large collection of homologous cycles that are
pairwise crossing in a graph embedded in the projective plane.

\subsection{Representing isometric cycles in each homology class}\label{sec:repr-isom-cycl}

We begin by determining the homology classes of each of the $O(2^{2g} n)$ isometric cycles in the
following manner.
Let~$p$ be a simple path, and let~$[p]_h$ denote the bitwise exclusive-or of the homology
signatures of its edges.
Let~$r$ be the root of any shortest path tree~$T$.
Recall, $\sigma(r,v)$ denotes the shortest path between $r$ and $v$.
It is straightforward to compute $[\sigma(r,v)]_h$ for every vertex $v \in V$ in~$O(gn)$ time by
iteratively computing signatures in a leafward order.
Then, the homology signature of any isometric cycle~$\gamma = \sigma(r,u) \cdot uv \cdot
\sigma(v,r)$ can be computed in~$O(g)$ time as $[\sigma(r,u)]_h \oplus [uv]_h \oplus
[\sigma(r,v)]_h$.
We spend~$O(2^{2g} g n) = O(2^{2g} n^2)$ time total computing homology signatures and therefore
homology classes.
For the remainder of this section, we consider a set of isometric cycles $\cal C$ in a single homology class.  

Let $\gamma, \gamma' \in {\cal C}$ be two isometric cycles in the same homology class.
The combination $\gamma \oplus \gamma'$ forms the boundary of a subset of faces.
That is, $(G \snip \gamma) \snip  \gamma'$ contains at least two components.
We represent the cycles in ${\cal C}$ by a tree $T_{\cal C}$ where each edge~$e$ of~$T_{\cal C}$
corresponds to a cycle~$\gamma(e) \in {\cal C}$ and each node~$v$ in~$T_{\cal C}$ corresponds to a
subset $F(v) \in (F \cup \{f_\infty\})$;
specifically, the nodes correspond to sets of faces in the components of $G \snip {\cal C}$.
This tree generalizes the {\em region tree} defined by Borradaile, Sankowski and Wulff-Nilsen
for planar graphs~\cite{BSW15} to more general orientable surface embedded graphs.
We also designate a single representative cycle~$\gamma({\cal C})$ of~${\cal C}$ and pre-compute
its cycle signature~$[\gamma({\cal C})]$ for use in our basis cycle finding procedure.
See Figure~\ref{fig:rTrees}.

\begin{figure}[t]
  \centering
    \includegraphics[height=1.7in]{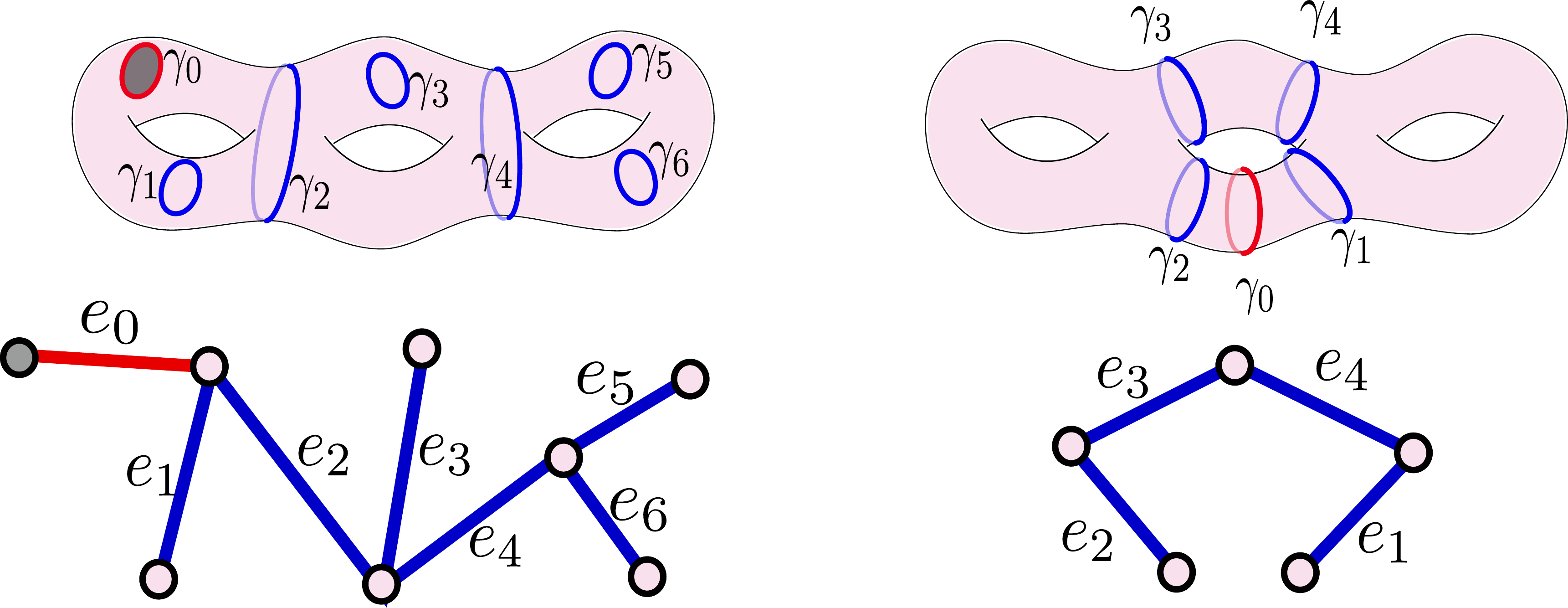}
  \caption{Two collections of homologous cycles and their generalized region trees.
    Left: The cycles are null-homologous.
  Right: The cycles lie in a non-trivial homology class.}
  \label{fig:rTrees}
\end{figure}

We describe here the construction of $T_{\cal C}$.
Initially,~$T_{\cal C}$ is a single vertex with one (looping) edge to itself (we will
guarantee~$T_{\cal C}$ is a tree later).
Let~$\gamma_0$ be an arbitrary cycle in~${\cal C}$.
We compute $G' = G \snip \gamma_0$.
For the one vertex~$v$ of~$T_{\cal C}$, we set~$F(v) = F \cup \{f_\infty\}$ and for the one
edge~$e$, we set~$\gamma(e) = \gamma_0$.

We maintain the invariants that every component of $G'$ is bound by at least two cycles of ${\cal
C}$
(initially the cycle~$\gamma_0$ is used twice), each vertex of $T_{\cal C}$ is associated with all
faces in one component of~$G'$ (possibly including $f_\infty$), and each edge~$e$ in $T_{\cal C}$
is associated with the cycle in ${\cal C}$ bounding the faces for the two vertices incident
to~$e$.
Assuming these invariants are maintained, and because cycles in ${\cal C}$ do not cross, each
cycle in ${\cal C}$ lies entirely within some component of~$G'$.
For each cycle $\gamma \in {\cal C} \setminus \Set{\gamma_0}$, we set $G' := G' \snip \gamma$,
subdivide the vertex associated with the faces of $C$'s component, associate the two sets of faces
created in~$G'$ with the two new vertices of $T_{\cal C}$, and associate the new edge of $T_{\cal
C}$ with $\gamma$.

Let~$r$ be the vertex of~$T_{\cal C}$ associated with~$f_{\infty}$.
If cycles in $\cal C$ have trivial homology, then they each separate~$G$, and $T_{\cal C}$ is a
tree.
We root~$T_{\cal C}$ at~$r$ and let~$\gamma({\cal C})$ be an arbitrary cycle.
Otherwise, let~$e$ be an arbitrary edge incident to~$r$.
We set~$\gamma({\cal C})$ to be $\gamma(e)$, remove~$e$ from~$T_{\cal C}$, and
root~$T_{\cal}$ at~$r$.
Observe that~$T_{\cal}$ has exactly one leaf other than~$r$ in this case.

Computing $G' \snip \gamma$ for one cycle~$\gamma$ takes~$O(n)$ time.
Therefore, we can compute~$T_{\cal C}$ in~$O(n^2)$ time total.

\subsection{Selecting an isometric cycle from a homology class}\label{sec:select-an-isom}

Let~$S$ be an $(m - n + 1)$-bit support vector.
We describe a procedure to compute $\langle S, [\gamma] \rangle$ for every isometric
cycle~$\gamma$ in~$G$
in~$O(2^{2g} n)$ time.
Using this procedure, we can easily return the minimum weight cycle such that $\langle S, [\gamma]
\rangle = 1$.

We begin describing the procedure for cycles in the trivial homology class.
Let~${\cal C}$ be the collection of null-homologous isometric cycles computed above, and
let~$T_{\cal C}$ be the tree computed for this set.
Consider any edge~$e$ of~$T_{\cal C}$.
The first~$2g$ bits of~$[\gamma(e)]$ are equal to~$0$, because any co-cycle crosses a cut in the
dual an even number of times.
Cycle~$\gamma(e)$ bounds a subset of faces~$F'$. 
In particular,~$F'$ is the set of faces associated with vertices lying \emph{below}~$e$ in~$T_{\cal
C}$.
The $i$th bit of~$[\gamma(e)]$ is~$1$ if and only if~$p_i$ crosses cut~$F'$ an odd number of
times; in other words, the $i$th bit is~$1$ if and only if~$f_i \in F'$.

We compute $\langle S, [\gamma] \rangle$ for every cycle~$\gamma \in {\cal C}$ in~$O(n)$ time
by essentially walking up~$T_{\cal C}$ in the following manner.
For each edge~$e$ in~$T_{\cal C}$ going to a leaf~$v$, we maintain a bit~$z$ initially equal
to~$0$ and iterate over each face~$f_i \in F(v)$.
If the~$i$th bit of~$S$ is equal to~$1$ then we flip~$z$.
After going through all the faces in~$F(v)$,~$z$ is equal to~$\langle S, [\gamma(e)] \rangle$.

We then iterate up the edges of~$T_{\cal C}$ toward the root.
For each edge~$e$, we let~$v$ be the lower endpoint of~$e$ and set bit~$z$ equal to the
exclusive-or over all $\langle S, \gamma(e') \rangle$ for edges~$e'$ lying below~$v$.
We then iterate over the faces of~$F(v)$ as before and set~$\langle S, [\gamma(e)] \rangle$ equal
to~$z$ as before.
We iterate over every face of~$G$ at most once during this procedure, so it takes~$O(n)$ time
total.

Now, consider the set of isometric cycles~${\cal C}$ for some non-trivial homology class.
Consider any edge~$e$ of~$T_{\cal C}$.
Let~$F'$ be the subset of faces bound by~$\gamma({\cal C}) \oplus \gamma(e)$.
The $i$th bit of~$[\gamma(e)]$ disagrees with the $i$th bit of~$[\gamma({\cal C})]$ if and only if
path~$p_i$ crosses dual cut~$F'$ an odd number of times; in other words, the $i$th bits differ if
and only if~$f_i \in F'$.
By construction,~$\gamma({\cal C})$ lies on the boundary of~$F(r)$ and~$F(v)$ where~$r$ and~$v$
are the root and other leaf of~$T_{\cal C}$ respectively.
Root~$r$ is the only node of~$T_{\cal C}$ associated with~$f_{\infty}$.
We conclude the $i$th bit of~$[\gamma(e)]$ disagrees with~$[\gamma({\cal C})]$ if and only
if~$f_i$ is associated with a vertex lying below~$e$ in~$T_{\cal C}$.

We again walk up~$T_{\cal C}$ to compute $\langle S, [\gamma] \rangle$ for every cycle~$\gamma \in
{\cal C}$.
Recall,~$[\gamma({\cal C})]$ is precomputed and stored with~$T_{\cal C}$.
For each edge~$e$ of~$T_{\cal C}$ in rootward order, let~$v$ be the lower endpoint of~$e$.
Let~$e'$ be the edge lying below $e$ in~$T_{\cal C}$ if it exists.
If~$e'$ does not exist, we denote~$\gamma(e')$ as~$\gamma({\cal C})$.
We set~$z$ equal to $\langle S, \gamma(e') \rangle$.
We then iterate over the faces of~$F(v)$ as before, flipping~$z$ once for every bit~$i$ where~$f_i
\in F(v)$ and bit~$i$ of~$S$ is equal to~$1$.
We set $\langle S, \gamma(e) \rangle := z$.
As before, we consider every face at most once, so the walk up~$T_{\cal C}$ takes~$O(n)$ time.

We have shown the following lemma, which concludes the discussion of our minimum cycle basis
algorithm.
\begin{lemma}
  \label{lem:min_odd_hit}
  Let~$G$ be a graph with~$n$ vertices,~$m$ edges, and~$\ell$ faces cellulary embedded in an
  orientable surface of genus~$g$ such that~$m = O(n)$ and~$\ell = O(n)$.
  We can preprocess~$G$ in~$O(2^{2g} n^2)$ time so that for any $(m - n + 1)$-bit support
  vector~$S$ we can compute the minimum weight cycle~$\gamma$ such that $\langle S, \gamma
  \rangle = 1$ in~$O(2^{2g} n)$ time.
\end{lemma}
\begin{theorem}
  \label{thm:cycle_basis}
  Let~$G$ be a graph with~$n$ vertices and~$m$ edges, cellularly embedded in an orientable surface of genus~$g$.
  We can compute a minimum weight cycle basis of~$G$ in~$O(n^{\omega} + 2^{2g} n^2 + m)$ time.
\end{theorem}

\subsection{Improving the time for computing and representing isometric cycles}\label{sec:impr-time-find}

Here we discuss ways in which we can improve the running time for finding and representing isometric cycles using known techniques, thereby isolating the bottleneck of the algorithm to updating the support vectors and computing $\gamma_j$.

The set and representation of isometric cycles can computed recursively using $O(\sqrt{gn})$ balanced separators (e.g.~\cite{AD96}) as inspired by Wulff-Nilsen~\cite{W09}.  Briefly, given a set $S$ of $O(\sqrt{gn})$ separator vertices (for a graph of bounded genus), find all the isometric cycles in each component of $G \setminus S$ and represent these isometric cycles in at most $2^{2g}$ region trees per component, as described above.  Merging the region trees for different components of $G \setminus S$ is relatively simple since different sets of faces are involved.  It remains to compute the set of isometric cycles that contain vertices of $S$ and add them to their respective region trees.  First note that a cycle that is isometric in $G$ and does not contain a vertex of $S$ is isometric in $G \setminus S$, but a cycle that is isometric in $G \setminus S$ may not be isometric in $G$, so indeed we are computing a superset of the set of isometric cycles via this recursive procedure.  However, it is relatively easy to show that an isometric cycle of $G \setminus S$ can cross an isometric cycle of $G$ at most once, so, within a given homology class, isometric cycles will nest and be representable by a region tree.

To compute the set of isometric cycles that intersect vertices of $S$, we first compute shortest
path trees rooted at each of the vertices of $S$, generating the Horton cycles rooted at these
vertices; this procedure takes $O(\sqrt{gn}\cdot n)$ time using the linear time shortest path
algorithm for graphs excluding minors of sub-linear size~\cite{TM09}.  We point out that the
algorithm of Amaldi \etal~\cite{AIJMR09} works by identifying Horton cycles that are not
isometric and by identifying, among different Horton-cycle representations of a given isometric
cycle, one representative; this can be done for a subset of Horton cycles, such as those rooted in
vertices of $S$, and takes time proportional to the size of the representation of the Horton
cycles (i.e., the $O(\sqrt{n})$ shortest path trees, or $O(\sqrt{g}n^{1.5})$).

For a given homology class of cycles, using the shortest-path tree representation of the isometric
cycles, we can identify those isometric cycles in that homology class by computing the homology
signature of root-to-node paths in the shortest path tree as before; this process can be done in $O(\sqrt{g}n^{1.5})$ time.  We must now add these cycles to the corresponding region tree. Borradaile, Sankowski and Wulff-Nilsen~\cite{BSW15} describe a method for adding $n$ cycles to a region tree in $O(n \poly \log n)$ time that is used in their minimum cycle basis algorithm for planar graphs; this method will generalize to surfaces for nesting cycles.  Therefore computing the homology classes of these isometric cycles and adding these isometric cycles to the region trees takes a total of $O(2^{2g}\sqrt{g}n^{1.5})$ time.

In total, this recursive method for computing and building a representation of a superset of the isometric cycles takes time given by the recurrence relation $T(n) = 2T(n/2)+O(2^{2g}\sqrt{g}n^{1.5})$ or $O(2^{2g}\sqrt{g}n^{1.5})$ time.

\section{Homology Basis}
\label{sec:homology_basis}
We now describe our algorithm for computing a minimum homology basis.
Our algorithm works for both orientable and non-orientable surfaces, although we assume without
loss of generality that the surface contains at least one boundary component.
At a high level, our algorithm for minimum homology bases is very similar to our algorithm for
minimum cycle bases.
As before, our algorithm incrementally adds simple cycles~$\gamma_1, \dots, \gamma_{\beta}$ to the
minimum homology basis by maintaining a set of~$\beta$ support vectors~$S_1, \dots, S_{\beta}$
such that the following hold:
\begin{itemize}
  \item
    The support vectors form a basis for~$\Z_2^{\beta}$.
  \item
    When the algorithm is about to compute the $j$th cycle~$\gamma_j$ for the minimum homology
    basis, $\langle S_j, [\gamma_{j'}]_h \rangle = 0$ for all~$j' < j$.
\end{itemize}
Our algorithm chooses for~$\gamma_j$ the minimum weight simple cycle~$\gamma$ such that~$\langle
S_j, [\gamma]_h \rangle = 1$.
The following lemma has essentially the same proof as Lemma~\ref{lem:min_weight_odd_hit}.
\begin{lemma}
  \label{lem:min_weight_odd_hit_homology}
  Let~$S$ be a $\beta$-bit vector with at least one bit set to~$1$, and let~$\eta$ be the
  minimum weight cycle such that~$\langle S, [\eta]_h \rangle = 1$.
  Then,~$\eta$ is a member of the minimum homology basis.
\end{lemma}

As before, our algorithm updates the support vectors and computes minimum homology basis cycles in
a recursive manner.
We define~$\extend(j,k)$ and~$\update(j,k)$ as before, using homology signatures in place of
cycle signatures when applicable.
Our algorithm runs $\extend(1,\beta)$ to compute the minimum homology basis.

The one crucial difference between our minimum cycle basis and minimum homology basis algorithms
is the procedure we use to find each minimum homology basis cycle~$\gamma_j$ given support
vector~$S_j$.
The homology basis procedure takes~$O(\beta^2 n \log n)$ time instead of~$O(2^{2g} n)$ time, and
it requires no preprocessing step.
We describe the procedure in Sections~\ref{subsec:double_cover}
and~\ref{subsec:homology_basis_cycles}.

The procedure~$\update(j,k)$ takes only~$O(\beta k^{\omega - 1})$ time in our minimum homology
basis algorithm, because signatures have length~$\beta$.
Therefore, we can bound the running time of $\extend(j,k)$ using the following recurrence:
$$
  T(k) =
  \begin{cases}
    2T(k/2) + O(\beta k^{\omega - 1}) & \text{if } k > 1 \\
    O(\beta^2 n \log n) & \text{if } k = 1
  \end{cases}
$$

The total time spent in calls to $\extend(j,k)$ where~$k > 1$ is~$O(\beta k^{\omega - 1})$.
The total time spent in calls to $\extend(j,1)$ is~$O(\beta^2 k n \log n)$.
Therefore,~$T(k) = O(\beta k^{\omega - 1} + \beta^2 k n \log n)$.
The running time of our minimum homology basis algorithm%
\footnote{Our minimum homology basis algorithm can be simplified somewhat by having~$\extend(j,k)$
  recurse on~$\extend(j,1)$ and~$\extend(j+1, k-1)$ and by using a simpler algorithm
  for~$\update(j,k)$.
  This change will increase the time spent in calls to~$\extend(j,k)$ where~$k > 1$, but the time
  taken by calls with~$k = 1$ will still be a bottleneck on the overall run time.}
  (after sparsifying~$G$) is~$T(\beta) = O(\beta^3 n \log n) = O((g + b)^3 n
  \log n)$.

\subsection{Cyclic double cover}
\label{subsec:double_cover}

In order to compute minimum homology basis cycle~$\gamma_j$, we lift the graph into a
\emph{covering space} known as the \EMPH{cyclic double cover}.
Our presentation of the cyclic double cover is similar to that of Erickson~\cite{e-sncds-11}.
Erickson describes the cyclic double cover relative to a single simple non-separating cycle in an
orientable surface;
however, we describe it relative to an arbitrary \emph{set} of non-separating \emph{co-paths}
determined by a support vector~$S$, similar to the homology cover construction of Erickson and
Nayyeri~\cite{en-mcsnc-11}.
Our construction works for non-orientable surfaces without any special considerations.

\begin{figure}[h]
  \centering
    \includegraphics[height=0.9in]{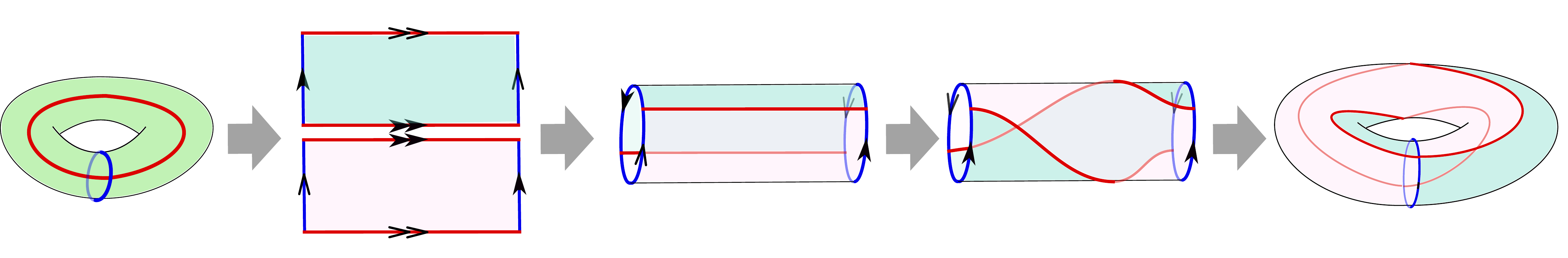}
    \caption{Constructing the cyclic double cover.
    Left to right:
    A pair of co-cycles~$\Psi$ on the torus~$\Sigma$;
    the surfaces~$(\Sigma', 0)$ and~$(\Sigma',1)$;
    identifying copies of one co-cycle;
    preparing to identify copies of the other co-cycle;
    the cyclic double cover.}
  \label{fig:dCover}
\end{figure}

Let~$S$ be a $\beta$-bit support vector for the minimum homology basis problem as defined above.
We define the cyclic double cover relative to~$S$ using a standard \emph{voltage
construction}~\cite[Chapter 4]{gt-tgt-01}.
Let~$G^2_S$ be the graph whose vertices are pairs~$(v,z)$, where~$v$ is a vertex of~$G$ and~$z$ is
a bit.
The edges of~$G^2_S$ are ordered pairs $(uv, z) := (u,z)(v, z \oplus \langle S, [uv]_h \rangle)$
for all edges~$uv$ of~$G$ and bits~$z$.
Let~$\pi : G^2_S \to G$ denote the \EMPH{covering map}~$\pi(v,z) = v$.
The \EMPH{projection} of any vertex, edge, or path in~$G^2_S$ is the natural map to~$G$ induced
by~$\pi$.
We say a vertex, edge, or path~$p$ in~$G$ \EMPH{lifts} to $p'$ if $p$ if the projection of $p'$.
A closed path in~$G^2_S$ is defined to bound a face (be a boundary component) of~$G^2_S$ if and
only if its projection with regard to~$\pi$ bounds a face (is a boundary component) of~$G$.
This construction defines an embedding of~$G^2_S$ onto a surface~$\Sigma^2_S$ (we will
prove~$G^2_S$ and~$\Sigma^2_S$ are connected shortly).

We can also define~$G^2_S$ in a more topologically intuitive way as follows.
Let~$\Psi$ be a set of co-paths which contains each co-path~$p_i$ for which the $i$th bit
of~$S$ is equal to~$1$.
Let~$\Sigma'$ be the surface obtained by cutting~$\Sigma$ along the image of each co-path
in~$\Psi$.
Note that~$\Sigma'$ may be disconnected.
Each co-path~$p_i \in \Psi$ appears as two copies on the boundary of~$\Sigma'$ denoted~$p_i^-$
and~$p_i^+$ (note that~$p_i^-$ and~$p_i^+$ may themselves be broken into multiple components).
Create two copies of~$\Sigma'$ denoted~$(\Sigma', 0)$ and~$(\Sigma', 1)$, and let~$(p_i^-,z)$
and~$(p_i^+,z)$ denote the copies of~$p_i^-$ and~$p_i^+$ in surface~$(\Sigma', z)$.
For each co-path~$p_i \in \Psi$, we identify~$(p_i^+, 0)$ with~$(p_i^-, 1)$ and we
identify~$(p_i^+, 1)$ with~$(p_i^-, 0)$, creating the surface~$\Sigma^2_S$ and the graph~$G^2_S$
embedded on~$\Sigma^2_S$.
See Figure~\ref{fig:dCover}.

The first three of the following lemmas are immediate.
\begin{lemma}
  \label{lem:projection}
  Let~$\gamma$ be any simple cycle in~$G$, and let~$s$ be any vertex of~$\gamma$.
  Then~$\gamma$ is the projection of a unique path in~$G^2_S$ from~$(s, 0)$ to~$(s, \langle S,
  [\gamma]_h \rangle)$.
\end{lemma}
\begin{lemma}
  \label{lem:short_lifts}
  Every lift of a shortest path in~$G$ is a shortest path in~$G^2_S$.
\end{lemma}
\begin{lemma}
  \label{lem:short_odd_projection}
  Let~$\gamma$ be the minimum weight simple cycle of~$G$ such that $\langle S, [\gamma]_h \rangle
  = 1$, and let~$s$ be any vertex of~$\gamma$.
  Then~$\gamma$ is the projection of the shortest path in~$G^2_S$ from~$(s,0)$ to~$(s,1)$.
\end{lemma}
\begin{lemma}
  \label{lem:double_cover_connected}
  The cyclic double cover~$G^2_S$ is connected.
\end{lemma}
\begin{proof}
  There exists some simple cycle~$\gamma$ in~$G$ such that $\langle S, [\gamma]_h \rangle = 1$.
  Let~$s$ be any vertex of~$\gamma$.
  Let~$v$ be any vertex of~$G$.
  We show there exists a path from~$(v,z)$ to~$(s,0)$ in~$G^2_S$ for both bits~$z$.
  There exists a path from~$v$ to~$s$ in~$G$ so there is a path from~$(v,z)$ to one of~$(s,0)$
  or~$(s,1)$ in~$G^2_S$.
  The other of $(s,0)$ or $(s,1)$ may be reached by following the lift of~$\gamma$.
\end{proof}

Observe that~$G^2_S$ has~$2n$ vertices and~$2m$ edges.
Each co-path~$p_i$ shares an even number of edges with each face of~$G$.
By Lemma~\ref{lem:projection}, both lifts of any face~$f$ to~$G^2_S$ are cycles; in particular
both lifts are faces.
However, there may be one or more boundary cycles~$\gamma$ of~$G$ such that~$\langle S, [\gamma]_h
\rangle = 1$.
It takes both lifts of such a cycle~$\gamma$ to make a single boundary component in~$G^2_S$.
We conclude~$G^2_S$ contains~$2\ell$ faces and between~$b$ and~$2b$ boundary cycles.
Surface~$\Sigma^2_S$ has Euler characteristic $2n - 2m + 2\ell = 2\chi$.
It is non-orientable if and only if there exists a one-sided cycle~$\eta$ such that~$\langle S,
[\eta]_h \rangle = 0$.
If both~$\Sigma$ and~$\Sigma^2_S$ are non-orientable, then~$\Sigma^2_S$ has genus at most~$2g+b$.
If only~$\Sigma$ is non-orientable, then~$\Sigma^2_S$ has genus at most~$g + b/2 - 1$.
If both surfaces are orientable, then~$\Sigma^2_S$ has genus at most~$2g + b/2 - 1$.
In all three cases, the genus is at most~$O(\beta)$.

\subsection{Selecting homology basis cycles}
\label{subsec:homology_basis_cycles}

Let~$S$ be any $\beta$-bit support vector.
We now describe our algorithm to select the minimum weight cycle~$\gamma$ such that~$\langle S,
[\gamma]_h \rangle = 1$.
Our algorithm is based on one by Erickson and Nayyeri~\cite{en-mcsnc-11} for computing minimum
weight cycles in arbitrary homology classes, except we use the cyclic double cover instead of
their $\Z_2$-homology cover.
We have the following lemma.
While it was shown with orientable surfaces in mind, the proof translates verbatim to the
non-orientable case.
\begin{lemma}[Erickson and Nayyeri~{\cite[Lemma 5.1]{en-mcsnc-11}}]
  In~$O(n \log n + \beta n)$ time, we can construct%
\footnote{We only need to construct~$\Pi$ once for the entire minimum homology basis
algorithm, but constructing it once per basis cycle does not affect the overall run time.}
  a set~$\Pi$ of~$O(\beta)$ shortest paths in~$G$, such that every
  non-null-homologous cycle in~$G$ intersects at least one path in~$\Pi$.
\end{lemma}
Let~$G^2_S$ be the cyclic double cover of~$G$ with regard to~$S$.
Our algorithm constructs~$G^2_S$ in~$O(\beta n)$ time.

Suppose our desired cycle~$\gamma$ intersects shortest path~$\sigma \in \Pi$ at some vertex~$s$.
By Lemma~\ref{lem:short_odd_projection}, simple cycle~$\gamma$ is the projection of the shortest
path in~$G^2_S$ from~$(s,0)$ to~$(s,1)$.
Let~$\hat{\gamma}$ be this shortest path in~$G^2_S$.
Let~$\hat{\sigma}$ be the lift of~$\sigma$ to~$G^2_S$ that contains vertex~$(s,0)$.
By Lemma~\ref{lem:short_lifts}, path~$\hat{\sigma}$ is also a shortest path in~$G^2_S$.
If~$\hat{\gamma}$ uses any other vertex~$(v,z)$ of~$\hat{\sigma}$ other than~$(s,0)$, then it can
use the entire subpath of~$\hat{\sigma}$ between~$(s,0)$ and~$(v,z)$.

Now, consider the surface~$\Sigma^2_S \snip \hat{\sigma}$ which contains a single face bounded by
two copies of~$\hat{\sigma}$ we denote~$\hat{\sigma}^-$ and~$\hat{\sigma}^+$.
For each vertex~$(v,z)$ on~$\hat{\sigma}$, let~$(v,z)^-$ and~$(v,z)^+$ denote its two copies
on~$\hat{\sigma}^-$ and~$\hat{\sigma}^+$ respectively.
From the above discussion, we see~$\hat{\gamma}$ is a shortest path in~$\Sigma^2_S \snip
\hat{\sigma}$ from one of~$(s,0)^-$ or~$(s,0)^+$ to~$(s,1)$.

To find~$\gamma$, we use the following generalization of Klein's~\cite{k-msspp-05} multiple-source
shortest path algorithm:
\begin{lemma}[Cabello~\etal~\cite{cce-msspe-13}]
  \label{lem:mssp}
  Let~$G$ be a graph with~$n$ vertices, cellularly embedded in a surface
  of genus~$g$, and let~$f$ be any face of~$G$.
  We can preprocess~$G$ in~$O(g n \log n)$ time and~$O(n)$
  space so that the length of the shortest path from any vertex incident to~$f$ to any other
  vertex can be retrieved in~$O(\log n)$ time.
\end{lemma}

Our algorithm iterates over the~$O(\beta)$ shortest paths present in~$\Pi$.
For each such path~$\sigma$, it computes a lift~$\hat{\sigma}$ in~$G^2_S$, cuts~$\Sigma^2_G$
along~$\hat{\sigma}$, and runs the multiple-source shortest path procedure of Lemma~\ref{lem:mssp}
to find the shortest path from some vertex~$(s,z)^{\pm}$ on~$\hat{\sigma}^{\pm}$ to~$(s,z \oplus
1)$.
Each shortest path it finds projects to a closed path~$\gamma'$ such that $\langle S, [\gamma']_h
\rangle = 1$.
By the above discussion, the shortest such projection can be chosen for~$\gamma$.
Running the multiple-source shortest path procedure~$O(\beta)$ times on a graph of
genus~$O(\beta)$ takes~$O(\beta^2 n \log n)$ time total.
We conclude the discussion of our minimum weight homology basis algorithm.
\begin{lemma}
  \label{lem:min_odd_homology}
  Let~$G$ be a graph with~$n$ vertices,~$m$ edges, and~$\ell$ faces cellulary embedded in a
  surface of genus~$g$ such that~$m = O(n)$ and~$\ell = O(n)$.
  For any $\beta$-bit support vector~$S$ we can compute the minimum weight cycle~$\gamma$ such that
  $\langle S, \gamma \rangle = 1$ in~$O(\beta^2 n \log n)$ time.
\end{lemma}
\begin{theorem}
  \label{thm:homology_basis}
  Let~$G$ be a graph with~$n$ vertices and~$m$ edges, cellularly embedded in an orientable or
  non-orientable surface of genus~$g$ with~$b$ boundary components.
  We can compute a minimum weight homology basis of~$G$ in~$O((g+b)^3 n \log n + m)$ time.
\end{theorem}

\bibliographystyle{abuser}
\bibliography{basis}

\end{document}